\documentclass[11pt]{article}
\usepackage{iftex}
\ifPDFTeX
\pdfoutput=1
\fi

\ifPDFTeX
  \usepackage{amsmath,mathtools}
  \usepackage{stix2}
  \usepackage[T1]{fontenc}
  \usepackage[utf8]{inputenc}

  \UndeclareTextCommand{\textasteriskcentered}{TS1}
  \DeclareTextSymbolDefault{\textasteriskcentered}{T1}
  \DeclareTextCommand{\textasteriskcentered}{T1}{\ensuremath{*}}

  \input{glyphtounicode}
\else
  \usepackage[math-style=TeX]{unicode-math}
  \usepackage{amsmath,mathtools}
\fi

\usepackage{tikz}
\usepackage{accsupp}
\usetikzlibrary{svg.path}
\definecolor{orcidlogocol}{HTML}{A6CE39}
\newcommand{\orcidlogo}{\BeginAccSupp{method=escape,Alt={ORCiD ID}}\begin{tikzpicture}[yscale=-1,transform shape]
\fill[orcidlogocol] svg{M256,128c0,70.7-57.3,128-128,128C57.3,256,0,198.7,0,128C0,57.3,57.3,0,128,0C198.7,0,256,57.3,256,128z};
    \fill[white] svg{M86.3,186.2H70.9V79.1h15.4v48.4V186.2z}
                 svg{M108.9,79.1h41.6c39.6,0,57,28.3,57,53.6c0,27.5-21.5,53.6-56.8,53.6h-41.8V79.1z M124.3,172.4h24.5c34.9,0,42.9-26.5,42.9-39.7c0-21.5-13.7-39.7-43.7-39.7h-23.7V172.4z}
                 svg{M88.7,56.8c0,5.5-4.5,10.1-10.1,10.1c-5.6,0-10.1-4.6-10.1-10.1c0-5.6,4.5-10.1,10.1-10.1C84.2,46.7,88.7,51.3,88.7,56.8z};
\end{tikzpicture}\EndAccSupp{}}
\newcommand\orcidlogosized[1]{\raisebox{-#1/5}{\resizebox{!}{#1}{\orcidlogo}}}
\newcommand\orcid[2][1em]{\mbox{\href{https://orcid.org/#2}{\orcidlogosized{#1}\hspace{\dimexpr #1/2}\nolinkurl{#2}}}}

\usepackage{microtype}
\usepackage{fullpage}

\makeatletter
\newcommand{\SwappedDeclarePairedDelimiter}[3]{%
  \expandafter\DeclarePairedDelimiter\csname Swapped\string#1\endcsname{#2}{#3}%
  \newcommand#1{%
    \@ifstar{\csname Swapped\string#1\endcsname}
            {\@ifnextchar[{\csname Swapped\string#1\endcsname}
                          {\csname Swapped\string#1\endcsname*}%
            }%
  }%
}
\makeatother

\SwappedDeclarePairedDelimiter{\of}{(}{)}

\DeclarePairedDelimiterX\setd[1]\{\}{%

#1
}

\def\set{\setd*}

\SwappedDeclarePairedDelimiter{\abs}{\lvert}{\rvert}

\SwappedDeclarePairedDelimiter{\norm}{\lVert}{\rVert}

\SwappedDeclarePairedDelimiter{\brof}{[}{]}

\SwappedDeclarePairedDelimiter{\floor}{\lfloor}{\rfloor}

\SwappedDeclarePairedDelimiter{\ceil}{\lceil}{\rceil}

\SwappedDeclarePairedDelimiter{\inner}{\langle}{\rangle}

\DeclarePairedDelimiterX{\diverg}[2]{(}{)}{#1\;\delimsize\|\;\mathopen{}#2}
\def\diver{\diverg*}

\DeclareMathOperator*{\Exp}{E}
\newcommand{\E}[2][]{\Exp_{#1}\brof{#2}}

\newcommand{\PR}[2][]{\Pr_{#1}\brof{#2}}
\DeclarePairedDelimiterX{\condpar}[2]{[}{]}{#1\,\delimsize\vert\,\mathopen{}#2}

\DeclarePairedDelimiterX{\condbr}[2]{[}{]}{#1\,\delimsize\vert\,\mathopen{}#2}
\newcommand{\CE}[3][]{\Exp_{#1}\condbr*{#2}{#3}}

\newcommand{\eps}{\varepsilon}

\ifcsname eqdef\endcsname
  \newcommand{\defeq}{\eqdef}
\else
  \ifPDFTeX
    \newcommand{\defeq}{\mathrel{\overset{\makebox[0pt]%
      {\mbox{\normalfont\tiny\sffamily def}}}{=}}}
  \else
    \newcommand{\defeq}{\symbol{"225D}}
  \fi
  
\fi

\usepackage{amsthm}

\usepackage[pdfencoding=unicode]{hyperref}

\DeclareMathOperator{\Var}{Var}

\usepackage[capitalize]{cleveref}

\theoremstyle{plain}
\newtheorem{thm}{Theorem}[section]
\Crefname{thm}{Theorem}{Theorems}
\newtheorem{prop}[thm]{Proposition}
\Crefname{prop}{Proposition}{Propositions}
\newtheorem{lem}[thm]{Lemma}
\Crefname{lem}{Lemma}{Lemmas}

\Crefname{cor}{Corollary}{Corollaries}

\newtheorem{conj}[thm]{Conjecture}

\theoremstyle{definition}
\newtheorem{defn}[thm]{Definition}

\theoremstyle{remark}
\newtheorem{rmrk}{Remark}

\numberwithin{equation}{section}

\newcommand{\V}[1]{V_{#1}}
\newcommand{\X}{X}
\newcommand{\B}{B}

\newcommand{\Eexp}[1]{\E{\exp\of{#1}}}

\DeclareMathOperator{\Binom}{Binomial}
\DeclareMathOperator{\divop}{D}

\DeclareMathOperator{\KL}{\divop}
\DeclareMathOperator{\BKL}{\divop}

\DeclareRobustCommand{\stirling}{\genfrac\{\}{0pt}{}}

\newcommand{\doi}[1]{\href{https://doi.org/#1}{\nolinkurl{DOI:#1}}}

\title{Finite-Sample Concentration of the Multinomial in Relative
  Entropy\footnote{Author's version of paper
    published in the IEEE Transactions on Information Theory,
    \doi{10.1109/TIT.2020.2996134}.}}
\author{Rohit Agrawal\thanks{
    Harvard John A. Paulson
    School of Engineering and Applied Sciences, Cambridge, MA 02138 USA.
    \orcid{0000-0001-5563-7402}.
    This work was supported in part by the National Science
    Foundation (NSF) under Grant CCF-1763299 to Salil Vadhan, and in part by
    the Department of Defense (DoD) through the National Defense Science and
    Engineering Graduate Fellowship (NDSEG) Program.
  }}
\hypersetup{%
  pdftitle={Finite-Sample Concentration of the Multinomial in Relative Entropy},
  pdfauthor={Rohit Agrawal},
  pdflang=en
}

\begin{document}
\maketitle
\begin{abstract}
  We show that the moment generating function of the Kullback--Leibler
  divergence (relative entropy) between the empirical distribution of
  $n$ independent samples from a distribution $P$ over a finite alphabet
  of size $k$ (i.e.~a multinomial distribution) and $P$ itself is no
  more than that of a gamma distribution with shape $k - 1$ and rate
  $n$. The resulting exponential concentration inequality becomes
  meaningful (less than 1) when the divergence $\eps$ is larger than
  $(k-1)/n$, whereas the standard method of types bound requires $\eps >
  \frac{1}{n} \cdot \log{\binom{n+k-1}{k-1}} \geq (k-1)/n \cdot \log(1 +
  n/(k-1))$, thus saving a factor of order $\log(n/k)$ in the standard
  regime of parameters where $n\gg k$. As a consequence, we also obtain
  finite-sample bounds on all the moments of the empirical divergence
  (equivalently, the discrete likelihood-ratio statistic),
  which are within constant factors (depending on the moment) of their
  asymptotic values.
  Our proof proceeds via a simple
  reduction to the case $k = 2$ of a binary alphabet (i.e.~a binomial
  distribution), and has the property that improvements in the case of
  $k = 2$ directly translate to improvements for general $k$. In
  particular, we conjecture a bound on the binomial moment generating
  function that would almost close the quadratic gap between our
  finite-sample bound and the asymptotic moment generating function
  bound from Wilks' theorem (which does not hold for finite samples).

\end{abstract}

\section{Introduction}

A key problem in statistics is to understand the rate of convergence of
an empirical distribution of independent samples to the true underlying
distribution. Indeed, this convergence is the basis of hypothesis
testing and statistical inference in general \cite{pitman_basic_1979}.
For the case of discrete distributions over a finite alphabet, the
Neyman--Pearson lemma \cite{ney_pea_problem_1933} shows that for optimal
hypothesis testing it is important to consider the
\emph{likelihood-ratio statistic}, or equivalently
\cite{har_tus_information_2012}, the Kullback--Leibler divergence
(relative entropy) from the true distribution to the empirical
distribution, as formally defined in \cref{intro_defn:multidiv}:

\begin{defn}\label{intro_defn:multidiv}
Let $\X = (\X_1, \dotsc, \X_k)$ be distributed according to a
multinomial distribution with $n$ samples and probabilities $P =
(p_1,\dotsc, p_k)$, and define
\[
  \V{n,k,P}
  \defeq
  \KL\diverg[\Big]{\of{\X_1/n,\dotsc, \X_k/n}}{\of{p_1, \dotsc, p_k}}
\]
where
\begin{align*}
  \KL\diverg[\Big]{\of{q_1,\dotsc, q_k}}{\of{p_1, \dotsc, p_k}}
  \defeq \sum_{i=1}^k q_i \log \frac{q_i}{p_i}
\end{align*}
is the Kullback--Leibler divergence between two probability
distributions on a finite set $\set{1, \dotsc, k}$ (represented as
probability mass functions), and $\log$ is in the natural base (as are
all logarithms and exponentials in this work). The
likelihood-ratio statistic is $2n\V{n,k,P}$ \cite{har_tus_information_2012}.
\end{defn}

In this language, the Neyman--Pearson lemma states that the uniformly most powerful
hypothesis test for significance $\alpha$ rejects a hypothesis
$P=(p_1,\dots,p_k)$ if and only if $\V{n,k,P}$ is at least
$\eps_\alpha$, where $\eps_\alpha$ is such that
$\PR{\V{n,k,P}\geq\eps_{\alpha}}\leq\alpha$.  To apply this test in
practice an upper bound on $\eps_{\alpha}$ is needed, so to maximize the
power of a provably correct finite-sample test we seek upper bounds on
$\PR{\V{} \geq \eps}$ which are meaningful (less than $1$) for $\eps$ as
small as possible. Equivalently, tight control on $\eps$ reduces the
number of samples needed to obtain a given level of significance, which
is of importance in areas as disparate as high-dimensional statistics
\cite{wainwright_highdimensional_2019}, combinatorial constructions in
complexity theory \cite{agrawal_samplers_conference}, and private
machine learning \cite{dia_wan_cal_san_robustness_2020}.

In this work, we focus on tail bounds for $\PR{\V{n,k,P}\geq\eps}$ which decay
exponentially for small $\eps$, ideally when $\eps\approx
\E{\V{n,k,P}}$. Paninski \cite{paninski_estimation_2003} showed that
$\E{\V{n,k,P}} \leq \log\of{1 + \frac{k-1}n} \leq \frac{k-1}n$, and
conversely Jiao et al.\ \cite{jia_ven_han_wei_maximum_2017} showed that
for $P$ the uniform distribution and large enough $n$ that
$\E{\V{n,k,U_k}} \geq \frac{k-1}{n}\cdot \frac12$, so in general the
smallest $\eps$ for which one can expect a meaningful bound is of order
$(k-1)/n$.  In this work, we derive the first tail bound decaying
exponentially in $\eps$ for $\eps$ as small as $(k-1)/n$, whereas existing
bounds either require $\eps$ to be at least order $(k-1)/n\cdot
\log(n/k)$ when $k < n$
(\cite{csiszar_method_1998,mar_jia_tan_now_wei_concentration_2019}) or
work only for the uniform distribution and
decay exponentially in $\eps^2$ (\cite{ant_kon_convergence_2001}), which
when $\eps < 1$ is significantly weaker than decay in $\eps^1$. Formally,
our result is as follows:
\begin{thm}\label{intro_thm:tail}
  Let $\V{n,k,P}$ be as in \cref{intro_defn:multidiv}. Then for all $\eps >
  \frac{k-1}{n}$, it holds that
  \[
    \PR{\V{n,k,P} \geq \eps} \leq e^{-n\eps} \cdot
    \of{\frac{e\eps n}{k - 1}}^{k-1}.
  \]
\end{thm}
\Cref{intro_thm:tail} is in fact an immediate corollary of our main
technical result, which is a bound on the moment generating function of
$\V{n,k,P}$.
\begin{thm}\label{intro_thm:mgf}
  Let $\V{n,k,P}$ be as in \cref{intro_defn:multidiv}. Then for all
  $0\leq t < n$ it holds that
  \[
    \E{\exp\of[\big]{t\cdot \V{n,k,P}}} \leq
    \of{\frac{1}{1-t/n}}^{k-1}.
  \]
\end{thm}
Note that this is also the moment generating function of a gamma
distribution with shape $k - 1$ and rate $n$. Bounding the moment
generating function is a standard technique to obtain concentration
bounds (see e.g.~\cite{bou_lug_mas_concentration_2013}), but to the best
of our knowledge \cref{intro_thm:mgf} is the first to give a finite
bound on $\E{\exp\of{s\cdot 2n\V{n,k,P}}}$ independent of $n$ for any
constant $s > 0$. As a consequence, we are able to give the first (to
the best of our knowledge) upper bounds on the $m$'th moments of
$2n\V{n,k,P}$ which do not depend on $n$ for all $m>2$. 
Using Wilks' theorem \cite{wilks_largesample_1938} on the asymptotic
distribution of the likelihood-ratio statistic, we are then able to
compute the asymptotic moments of $2n\V{n,k,P}$ for fixed $k$ and $P$ as
$n$ goes to infinity. Furthermore, our finite sample bounds on the
$m$'th non-central moment are within constant factors (with the constant
depending on $m$) of the asymptotic value.

The rest of this work is organized as follows.  In
\cref{sec:finitesampleproof} we prove
\cref{intro_thm:mgf,intro_thm:tail}, with the proof divided into two
parts: in \cref{sec:reduction} we show \cref{intro_thm:mgf} can be
derived from bounds for the special case of a binary alphabet ($k=2$),
e.g.~a binomial distribution, and in \cref{sec:binomial} we give a bound
for this simpler case.  In \cref{sec:asymptotics} we use
\cref{intro_thm:mgf} to derive moment bounds and asymptotic results.
Finally, in \cref{sec:discussion} we compare our bounds to
existing results in the literature and suggest possible directions
for future research, and in particular conjecture an improvement to
\cref{intro_thm:mgf} which would nearly close the quadratic gap between
our finite-sample bound and the bound of Wilks' theorem on the
asymptotic distribution of likelihood-ratio statistic (which does not
hold in general for finite $n$).

\section{Proof of Finite-Sample Bounds}
In this section we prove our main technical result, the moment
generating function bound of \cref{intro_thm:mgf}, and use it
to derive our new tail bound \cref{intro_thm:tail}.

\label{sec:finitesampleproof}
\subsection{Reducing the Multinomial to the Binomial}
\label{sec:reduction}
We first show that the moment generating function of the
empirical relative entropy for arbitrary finite alphabets of size $k$
can be bounded in terms of the special case $k = 2$. Formally, this
requires the bound to be of a particular form:
\begin{defn}\label{defn:sampindbound}
  A function $f:[0,1) \to \mathbb R$ is a \emph{sample-independent MGF
  bound for the binomial KL} if for every positive integer $n$, real
  $t\in [0, n)$, and $p\in [0,1]$ it holds that
  \[
    \Eexp{t\cdot \V{n,2,(p,1-p)}} \leq f(t/n).
  \]
\end{defn}
\begin{rmrk}
  Recalling that $2n\V{n,k,P}$ is the likelihood-ratio statistic,
  \cref{defn:sampindbound} is equivalent to requiring bounds on the
  moment generating function $\E{\exp\of{s\cdot 2n\V{n,2,(p,1-p)}}}$ for
  $0\leq s <1/2$ which do not depend on $n$ or $p$.
\end{rmrk}

We can now state our reduction.
\begin{prop}\label{thm:multi_bin_reduction}
  Let $P = (p_1,\dotsc, p_k)$ be a distribution on a set of size $k$ for
  $k \geq 2$. Then for every sample-independent MGF bound for the
  binomial KL $f:[0, 1) \to \mathbb R$ and $0\leq t < n$, the moment
  generating function of $\V{n,k,P}$ satisfies
  \[
    \Eexp{t\cdot \V{n,k,P}} \leq f(t/n)^{k-1}.
  \]
\end{prop}
\begin{proof}
  This is a simple induction on $k$. The base case $k = 2$ holds
  by definition of sample-independent MGF bound for the binomial KL.
  
  For the inductive step, we compute conditioned on the value of
  $\X_k$. Note that if $p_k = 1$ then the inductive step is trivial
  since $\V{n,k,P}=0$ with probability $1$, so assume that $p_k < 1$.
  For each $i \in \set{1, \dotsc, k-1}$ define $p'_i = p_i/(1-p_k)$,
  so that conditioned on $\X_k = m$, the variables $(\X_1, \dotsc,
  \X_{k-1})$ are distributed multinomially with $n - m$ samples
  and probabilities $P'=(p'_1,\dotsc, p'_{k-1})$. Simple rearranging
  (using the chain rule) implies that
  \begin{align}
    \V{n,k,P}
    &= \KL\diver{\of{\X_1/n,\dotsc,\X_k/n}}{\of{p_1,\dotsc,p_n}}
    \nonumber
    \\
    &= \KL\diverg[\big]{\of{\X_k/n, 1 - \X_k/n}}{\of{p_k,1-p_k}}
    + \frac{n - \X_k}{n}\cdot \V{n-\X_k,k-1,P'}
    \label{eqn:decomp}
  \end{align}
  where
  \[
    \V{n-\X_k,k-1,P'}
    =
    \KL\diver{\of{\frac{\X_1}{n-\X_k},\dotsc,
      \frac{\X_{k-1}}{n-\X_k}}}{\of{p'_1,\dotsc,p'_{k-1}}}
  \]
  and where we treat the second term of \cref{eqn:decomp}
  as $0$ if $\X_k = n$. Now for every $0\leq t < n$ we have
  \begin{align*}
    \Exp\brof[\bigg]{&\exp\of[\big]{t\cdot \V{n,k,P}}} \\
    &=\Exp\brof[\bigg]{\Exp\condbr[\Big]{\exp\of[\big]{t\cdot \V{n,k,P}}}{\X_k}}\\
    &= \Exp\brof[\Bigg]{\exp\of[\Big]{t\cdot \KL\diverg[\big]{\of{\X_k/n,1-\X_k/n}}{\of{p_k,1-p_k}}}
     \cdot \CE{\exp\of{
      t \cdot\frac{n-\X_k}{n}\cdot\V{n-\X_k,k-1,P'}}}{\X_k}}.
  \intertext{%
  Since $0 \leq t\cdot \frac{n - \X_k}{n} < n - \X_k$, the inductive hypothesis for
  $\V{n-\X_k,k-1,P'}$ implies the upper bound
  }
    &\leq \E{\exp\of[\Big]{t\cdot \KL\diverg[\big]{\of{\X_k/n,1-\X_k/n}}{\of{p_k,1-p_k}}}
      \cdot f\of{\frac{t(n-\X_k)/n}{n-\X_k}}^{k-2}}\\
    &= f(t/n)^{k-2} \cdot
    \E{\exp\of[\Big]{t\cdot \KL\diverg[\big]{\of{\X_k/n,1-\X_k/n}}{\of{p_k,1-p_k}}}}.
  \end{align*}
  By definition of a sample-independent MGF bound for the binomial KL,
  the second term is at most $f(t/n)$, so we get a bound of
  $f(t/n)^{k-1}$ as desired.
\end{proof}
\begin{rmrk}
  Mardia et al.\ \cite{mar_jia_tan_now_wei_concentration_2019} use the
  same chain rule decomposition of the multinomial KL to inductively
  bound the (non-exponential) moments.
\end{rmrk}

\subsection{Bounding the Binomial}
\label{sec:binomial}

It remains to give a sample-independent MGF bound for the binomial KL:
\begin{prop}\label{lem:binmgf}
  The function
    \[f(x) = \frac{1}{1-x} \]
  is a sample-independent MGF bound for the binomial KL.
\end{prop}
\begin{rmrk}
  Hoeffding's inequality \cite{hoeffding_probability_1963} can be used
  to give a simple proof of the weaker claim that $2^x/(1-x)$ is a
  sample-independent MGF bound for the binomial KL.
\end{rmrk}
\begin{proof}
  Let $\B_{n,p}$ denote a random variable with $\Binom(n, p)$ distribution.
  Using the fact that
  \[\exp{\of[\Big]{n\cdot \KL\diverg[\big]{\of{i/n,1-i/n}}{\of{p,1-p}}}}
    = \frac{\PR{\B_{n,i/n}=i}}{\PR{\B_{n,p}=i}}\]
  for any integers $0\leq i\leq n$, we can expand the moment generating
  function as
  \begin{align*}
    \E{\exp\of{nx\cdot \KL\diver{\of{\frac{\B_{n,p}}{n},1-\frac{\B_{n,p}}{n}}}{\of{p,1-p}}}}
    &= \sum_{i=0}^n \PR{\B_{n,p}=i}^{1-x}\PR{\B_{n,i/n}=i}^x.
  \end{align*}
  For every $n$ and $i$, the function $q \mapsto \PR{\B_{n,q}=i}=\binom ni q^i(1-q)^{n-i}$ is
  easily seen to be log-concave over $[0,1]$, so we can upper bound the moment generating function by
  \begin{align*}
    G_n(p,x)\defeq \sum_{i=0}^n \PR{\B_{n,(1-x)p+ix/n}=i}
    &= \sum_{i=0}^n \binom{n}{i}\of[\big]{(1-x)p + ix/n}^i\of[\big]{1 - \of[\big]{(1-x)p+ix/n}}^{n-i}
  \end{align*}
  It turns out $G_n$ does not depend on $p$ and can be simplified
  significantly, which we prove in the following two lemmas.
\begin{lem}
\label{lem:constant-p}
  For all non-negative integers $n$ and real numbers $x$ and $p$ we have
  $G_n(p,x)=G_n(0,x)$.
\end{lem}
\begin{proof}
  Define $R_n(q, x) = \sum_{i=0}^n \binom ni \of{q +
    ix/n}^i\of{1-q - ix/n}^{n-i}$ (where when $i=n=0$ we treat
    $0/0=1$) so that $G_n(p,x) = R_n\of{(1-x)p, x}$
    and it suffices to prove that $R_n(q,x)=R_n(0,x)$.
  We prove this by induction on $n$: the
    base case of $n = 0$ holds since $R_n(q, x) = 1$ always, and
    for the inductive step we have
  \begin{align*}
    \frac\partial{\partial q}R_n(q,x)
    &= \sum_{i=0}^n\binom ni \frac\partial{\partial q}\of{(q + ix/n)^i(1-q-ix/n)^{n-i}}\\
    &= \sum_{i=0}^n\binom ni \of{i(q + ix/n)^{i-1}(1-q-ix/n)^{n-i} - (n-i)(q + ix/n)^i(1-q-ix/n)^{n-i-1}}\\
    &= n\sum_{i=1}^n \binom{n-1}{i-1}\of{q + x/n + \frac{i-1}{n-1}\cdot \frac{x(n-1)}{n}}^{i-1}\of{1-q-x/n-\frac{i-1}{n-1}\cdot \frac{x(n-1)}{n}}^{n-1-(i-1)}\\
    &\qquad- n\sum_{i=0}^{n-1}\binom{n-1}{i}\of{q + \frac{i}{n-1}\cdot \frac{x(n-1)}{n}}^i\of{1-q- \frac{i}{n-1}\cdot \frac{x(n-1)}{n}}^{n-1-i}\\
    &= n\sum_{i=0}^{n-1} \binom{n-1}{i}\of{q + x/n + \frac{i}{n-1}\cdot \frac{x(n-1)}{n}}^{i}\of{1-q-x/n-\frac{i}{n-1}\cdot \frac{x(n-1)}{n}}^{n-1-i}\\
    &\qquad- n\sum_{i=0}^{n-1}\binom{n-1}{i}\of{q + \frac{i}{n-1}\cdot \frac{x(n-1)}{n}}^i\of{1-q- \frac{i}{n-1}\cdot \frac{x(n-1)}{n}}^{n-1-i}\\
    &= n\of{R_{n-1}\of{q + \frac xn, \frac{x(n-1)}{n}} - R_{n-1}\of{q, \frac{x(n-1)}{n}}}\\
    &= n\of[\big]{R_{n-1}\of{0, x(n-1)/n} - R_{n-1}\of{0,x(n-1)/n}}=0
  \end{align*}
  where the last line is by the inductive hypothesis.
\end{proof}
\begin{lem}
\label{lem:bound-zero-p}
  For all non-negative integers $n$ we have
    $G_n(p,x)=\displaystyle\sum_{i=0}^n\frac{n!}{n^i(n-i)!}\cdot x^i$.
\end{lem}
\begin{proof}
  By Lemma~\ref{lem:constant-p} we have that
  $G_n(p,x)=G_n(0,x)=\sum_{i=0}^n\of{\frac{ix}{n}}^i\of{1-\frac{ix}{n}}^{n-i}$
  is a polynomial in $x$ of degree at most $n$.  For any non-negative
  integer $i\leq n$ we can compute the coefficient of $x^i$ in
  $G_n(0,x)$ by summing over the power of $x$ contributed by the $(jx/n)^j$ term
  for each $j$:
  \begin{align*}
    \sum_{j=0}^i \binom{n}{j} \of{\frac{j}{n}}^j \cdot \binom{n-j}{i-j}\of{-\frac{j}{n}}^{i-j}
    &=
    \sum_{j=0}^i \frac{n!}{j!(n-j)!}\cdot \frac{(n-j)!}{(i-j)!(n-i)!}\cdot  \of{\frac{j}{n}}^i (-1)^{i-j}\\
    &=
    \frac{n!}{n^i(n-i)!}\cdot \frac{1}{i!}\sum_{j=0}^i \binom{i}{j} j^i (-1)^{i-j}
  \end{align*}
  where $\frac{1}{i!}\sum_{j=0}^i \binom{i}{j} j^i (-1)^{i-j}$
  is by definition the Stirling number of the second kind $\stirling ii$
  and is equal to $1$ (see e.g.~\cite[Chapter 6.1]{gra_knu_pat_concrete_1994}),
  so that we can simplify this to
  \[
    \frac{n!}{n^i(n-i)!}
  \]
  as desired.
\end{proof}
Putting together Lemma~\ref{lem:constant-p} and
Lemma~\ref{lem:bound-zero-p}, we have that the moment generating
function is at most
$G_n(p,x)=\sum_{i=0}^n \frac{n!}{n^i(n-i)!}x^i$, where
$\frac{n!}{n^i(n-i)!} = \prod_{j=0}^{i-1}\of{1-j/n}\leq 1$
and thus for each $x\in [0,1)$ we have $G_n(p,x)\leq\sum_{i=0}^n x^i
\leq \sum_{i=0}^\infty x^i = 1/(1-x)$.
\end{proof}

Together, \cref{thm:multi_bin_reduction,lem:binmgf} imply our
moment generating function bound (\cref{intro_thm:mgf}), and thus a
Chernoff bound implies our tail bound:
\begin{proof}[Proof of \cref{intro_thm:tail}]
  By \cref{intro_thm:mgf}, we know for every $t\in [0,n)$ that
  $\Eexp{t\cdot \V{n,k,P}} \leq \of{\frac1{1-t/n}}^{k-1}$, so by a
  Chernoff bound
  \begin{align*}
    \PR{\V{n,k,P} \geq \eps}
    &\leq
    \inf_{t\in [0, n)} \exp\of{-t\eps} \cdot \of{\frac1{1-t/n}}^{k-1}.
  \end{align*}
  The result follows by making the optimal choice $t/n = 1 -
  (k-1)/(\eps n)$ when $\eps > \of{k-1}/{n}$.
\end{proof}

\section{Moment and Asymptotic Bounds}
\label{sec:asymptotics}

In this section we use \cref{intro_thm:mgf} to give finite-sample
and asymptotic bounds on the moments of $\V{n,k,P}$. We will need
some basic facts about subexponential random variables, for which we
follow the textbook of Vershynin \cite{vershynin_highdimensional_2018}.

\begin{lem}[{\cite[Definition 2.7.5, Proposition 2.7.1]{vershynin_highdimensional_2018}}]
\label{lem:subexponential}
  There is a universal constant $C > 0$ such that every real-valued
    random variable $X$ with finite \emph{subexponential norm}
    $\norm{X}_{\psi_1} \defeq \inf\set{t>0:\E{\exp\of{\abs{X}/t}}\leq 2}$
    satisfies $\E{\abs{X}^m}^{1/m}\leq Cm \norm{X}_{\psi_1}$ for all $m\geq 1$.
\end{lem}

\Cref{lem:subexponential} allows us to bound the moments of $2n\V{n,k,P}$
uniformly for all $n$.
\begin{thm}
\label{cor:moments}
  For every $n$, $k$, and $P$, it holds that
  $\norm{2n\V{n,k,P}}_{\psi_1}\!\leq 4(k-1)$ and
  $\norm{2n\V{n,k,P}-\E{2n\V{n,k,P}}}_{\psi_1}\allowbreak\leq 8(k-1)$.
  In particular, there exist universal constants $C_1,C_2>0$ such that for all
  $n$, $k$, $P$ and $m\geq 1$
  \begin{align*}
    \E{\of{2n\V{n,k,P}}^m}&\leq (C_1m(k-1))^m&
    \E{\of{2n\V{n,k,P}-\E{2n\V{n,k,P}}}^m}&\leq (C_2m(k-1))^m
  \end{align*}
\end{thm}
\begin{proof}
  \Cref{intro_thm:mgf} implies for all $n$, $k$, and $P$ that
  $\E{\exp\of{\frac1{4(k-1)} \cdot 2n\V{n,k,P}}} \leq \of{1-\frac1{2(k-1)}}^{-(k-1)}
  \leq 2$, so by \cref{lem:subexponential} we have that
  $\norm{2n\V{n,k,P}}_{\psi_1}\leq 4(k-1)$.
  By the triangle inequality and convexity of norms, this lets us bound
  the norm of the centered random variable as
  $\norm{2n\V{n,k,P}-\E{2n\V{n,k,P}}}_{\psi_1}\leq
  2\norm{2n\V{n,k,P}}_{\psi_1}\leq 8(k-1)$.
\end{proof}

Our asymptotic results rely on Wilks' theorem \cite{wilks_largesample_1938}
on the asymptotic behavior of the likelihood ratio test, which for fixed $k$
and $P$ implies that the random variable $2n\V{n,k,P}$ converges in distribution to the
chi-squared distribution with $k-1$ degrees of freedom as $n$ goes to
infinity (see also \cite[Theorem 4.2]{csi_shi_information_2005}).
Though in general convergence in distribution does not imply
convergence of moments or of the moment generating function
\cite{billingsley_convergence_1999}, it turns out that the bounds
from \cref{cor:moments} are strong enough for convergence in distribution to
imply convergence of the moments.
\begin{thm}\label{thm:asymptotic}
  Let $k\geq 2$ be an integer and $P=(p_1,\dots,p_k)$ be a probability
  distribution over a finite alphabet of size $k$ with $p_i\neq 0$ for
  every $i\in\set{1,\dots,k}$. Then for every $m\geq 1$ we have
  \begin{align*}
    \lim_{n\to\infty} \E{(2n\V{n,k,P})^m}
    &= \E{\of{\chi^2_{k-1}}^m}
    =2^m\frac{\Gamma\of{m+\frac{k-1}2}}{\Gamma\of{\frac{k-1}2}}\\
    \lim_{n\to\infty} \E{\of{2n\V{n,k,P}-\E{2n\V{n,k,P}}}^m}
    &= \E{\of{\chi^2_{k-1}-\E{\chi^2_{k-1}}}^m}
  \end{align*}
  and for every $s\in[0,1/2)$ we have
  \begin{align*}
    \lim_{n\to\infty} \E{\exp\of{s\cdot 2n\V{n,k,P}}}
    &= \E{\exp\of{s\cdot \chi^2_{k-1}}}=\of{1-2s}^{-(k-1)/2}\\
    \lim_{n\to\infty} \E{\exp\of{s\cdot \of{2n\V{n,k,P}-\E{2n\V{n,k,P}}}}}
    &= \E{\exp\of{s\cdot \of{\chi^2_{k-1}-\E{\chi^2_{k-1}}}}}
    =e^{-(k-1)s}(1-2s)^{-(k-1)/2}
  \end{align*}
\end{thm}
\begin{rmrk}
  \cite{mar_jia_tan_now_wei_concentration_2019} prove the one-sided
  lower bound that $\liminf_{n\to\infty} \Var\of{2n\V{n,k,P}}
  \geq \Var\of{\chi^2_{k-1}}$, which is a special case of the
  second equality above.
\end{rmrk}
\begin{proof}
  Given a sequence of random variables $\of{X_n}_{n\in\mathbb N}$
  which convergence in distribution to a random variable $X$,
  a sufficient condition for $\lim_{n\to\infty} \E{X_n} = \E{X}$
  is that $\sup_n \E{\abs{X_n}^{1+\alpha}}<\infty$ for some
  $\alpha > 0$ (see e.g.~\cite{billingsley_convergence_1999}).

  Wilks' theorem \cite{wilks_largesample_1938} shows that
  $2n\V{n,k,P}$ converges in distribution to $\chi^2_{k-1}$,
  and thus the continuous mapping theorem implies that
  $\of{2n\V{n,k,P}}^m$ converges in distribution to
  $\of{\chi^2_{k-1}}^m$ for every $m\geq 1$.
  \Cref{cor:moments} implies
  $\sup_n \E{\abs{\of{2n\V{n,k,P}}^m}^2} \leq (Cm(k-1))^{2m}<\infty$,
  which establishes the first claim.
  In particular, for $m=1$ we have $\lim_{n\to\infty}\E{2n\V{n,k,P}}
  =\E{\chi^2_{k-1}}$, so Slutsky's theorem implies that
  $2n\V{n,k,P}-\E{2n\V{n,k,P}}$ converges in distribution to
  $\chi^2_{k-1}-\E{\chi^2_{k-1}}$.
  Again by the continuous mapping theorem we thus have that
  $\of{2n\V{n,k,P}-\E{2n\V{n,k,P}}}^m$ converges in distribution
  to $\of{\chi^2_{k-1}-\E{\chi^2_{k-1}}}^m$, so since
  \Cref{cor:moments} implies
  $\sup_n \E{\abs{\of{2n\V{n,k,P}-\E{2n\V{n,k,P}}}^m}^2}
  \leq (Cm(k-1))^{2m}<\infty$, we also get the second claim.

  For the moment generating function claims, first note that
  they are trivial for $s=0$, as both sides are always $1$,
  and for $s\in(0,1/2)$ we have $1/2>1/4+s/2>s$. Now, since
  the continuous mapping theorem implies $\exp(s\cdot 2n\V{n,k,P})$
  converges in distribution to $\exp\of{s\cdot \chi^2_{k-1}}$,
  and \cref{intro_thm:mgf} implies
  $\sup_n \E{\abs{\exp\of{(1/4+s/2)\cdot 2n\V{n,k,P}}}}
  \leq \of{1/2-s}^{k-1}<\infty$, we get the third claim.
  Finally, for the last claim, we again have that
  $\exp\of{s\cdot\of{2n\V{n,k,P}-\E{2n\V{n,k,P}}}}$ converges in
  distribution to $\exp\of{s\cdot\of{\chi^2_{k-1}-\E{\chi^2_{k-1}}}}$
  by the continuous mapping theorem, and since $\V{n,k,P}\geq 0$
  we have
  $\exp\of{(1/4+s/2)\cdot\of{2n\V{n,k,P}-\E{2n\V{n,k,P}}}}
  \leq \exp\of{(1/4+s/2)\cdot 2n\V{n,k,P}}$ and we conclude
  as for the third claim.
\end{proof}

\section{Discussion}
\label{sec:discussion}

In this section we compare our bounds to existing results in the
literature and discuss possible directions for future work.

\subsection{Moment generating function bounds}

To the best of our knowledge, this work is the first to explicitly
consider the moment generating function of the empirical divergence, and
existing tail bounds do not give finite bounds on $\sup_n
\E{\exp\of{x\cdot n\V{n,k,P}}}=\sup_n\int_0^\infty
\PR{n\V{n,k,P}>\frac{\log t}x}\,dt$ for any $k\geq 3$ or constant $x>0$.
Thus, we focus on comparing our finite sample bound (\cref{intro_thm:mgf})
to the asymptotic one (\cref{thm:asymptotic}).

In \cref{thm:asymptotic} we showed for all $x\in[0,1)$ that
$\lim_{n\to\infty}\E{\exp\of{x\cdot n\V{n,k,P}}}=\of{1-x}^{-(k-1)/2}$,
whereas our finite sample bound of \cref{intro_thm:mgf} instead gave the
upper bound $\E{\exp\of{x\cdot n\V{n,k,P}}}\leq \of{1-x}^{-(k-1)}$, which is
quadratically worse. This loss arises from our binomial bound from
\cref{lem:binmgf} of $(1-x)^{-1}$ for the case $k=2$, where the correct
asymptotic bound is $\of{1-x}^{-1/2}$.  Unfortunately, it is \emph{not}
the case that this latter asymptotic bound holds for all $n$, $p$, and
$0\leq x<1$: indeed, this is violated even for $(n,p,x)=(2,1/2,1/2)$.
Nevertheless, we conjecture that \cref{lem:binmgf} can be improved to
something closer to the asymptotic bound:
\begin{conj}\label{conj:betterbinarymgf}
  The function
  \[
    f(x) = \frac{2}{\sqrt{1 - x}} - 1
  \]
  is a sample-independent MGF bound for the binomial KL.
\end{conj}
\begin{rmrk}
$1/\sqrt{1-x}\leq 2/\sqrt{1-x}-1\leq 1/(1-x)$
for all $x\in [0,1)$.
\end{rmrk}
\Cref{conj:betterbinarymgf} would follow from the following more natural
conjecture, which looks at a single branch of the KL divergence and is
supported by numerical evidence:
\begin{conj}\label{conj:halfchisquared}
  Letting
  \[
    \BKL_{>}\diver{p}{q}
    \defeq
    \begin{cases}
      0&p\leq q\\
      \KL\diverg[\Big]{\of{p,1-p}}{\of{q,1-q}} &p>q
    \end{cases}
  \]
  it holds for every positive integer $n$, real $t\in [0,n)$, and $p\in
  [0,1]$ that
  \[
    \E{\exp\of[\Big]{t\cdot \BKL_{>}\diver{\B/n}{p}}} \leq \frac1{\sqrt{1-t/n}}
  \]
  where $\B \sim \Binom(n, p)$.
\end{conj}
\begin{rmrk}
  We believe the results (or techniques) of Zubkov and Serov
  \cite{zub_ser_complete_2013} and Harremo\"es
  \cite{harremoes_bounds_2017} strengthening Hoeffding's inequality may
  be of use in proving these conjectures.
\end{rmrk}

\begin{proof}[Proof of \cref{conj:betterbinarymgf} given \cref{conj:halfchisquared}]
  We have that
  \[
    \KL\diverg[\Big]{(p, 1-p)}{(q, 1-q)}
    = \BKL_{>}\diver pq + \BKL_{>}\diver{1-p}{1-q}
  \]
  so for every $i\in\set{0,1,\dotsc,n}$
  \[
    \exp\of[\Big]{t\cdot \KL\diverg[\Big]{(i/n,1-i/n)}{(p,1-p)}}
    =
    \exp\of[\Big]{t\cdot \BKL_{>}\diver{i/n}{p}}
    \cdot \exp\of[\Big]{t\cdot \BKL_{>}\diver{1-i/n}{1-p}}.
  \]
  Letting $x = \exp\of[\Big]{t\cdot \BKL_{>}\diver{i/n}{p}}$ and
  $y = \exp\of[\Big]{t\cdot \BKL_{>}\diver{1-i/n}{1-p}}$, we have
  that at least one of $x$ and $y$ is equal to $1$, so that
  \[
    xy = \of[\big]{1 + (x-1)}\of[\big]{1 + (y-1)} = 1 + (x-1) + (y-1)
    + (x-1)(y-1)
    = x + y - 1,
  \]
  and thus by taking expectations over $i = \B$ for $\B\sim \Binom(n,
  p)$, we get
  \[
    \Eexp{t\cdot \BKL\diver{\B/n}{p}}
    = \Eexp{t\cdot \BKL_{>}\diver{\B/n}{p}} +
      \Eexp{t\cdot \BKL_{>}\diver{1-\B/n}{1-p}}
      - 1.
    \]
  We conclude by bounding both terms using \cref{conj:halfchisquared},
  since $n-\B$ is distributed as $\Binom(n, 1-p)$.
\end{proof}

\subsection{Moment bounds}

The moments of $\V{n,k,P}$ have seen some study in the literature.  Most
notably, Paninski \cite{paninski_estimation_2003} showed by comparison
to the $\chi^2$-statistic that $\E{\V{n,k,P}} \leq \log\of{1 +
\frac{k-1}n} \leq \frac{k-1}n$. In the reverse direction,
\cite{jia_ven_han_wei_maximum_2017} showed that if $n\geq 15k$ then for
the uniform distribution it holds that $\E{\V{n,k,U_k}} \geq
\frac{k-1}{2n}$, complementing the asymptotic result that
$\lim_{n\to\infty}\E{n\V{n,k,U_k}} =\frac{k-1}{2}$, which follows from
\cref{thm:asymptotic} (and can also be derived from
\cite{mar_jia_tan_now_wei_concentration_2019}). For higher moments,
\cite{mar_jia_tan_now_wei_concentration_2019} showed that
$\Var(\V{n,k,P})\leq Ck/n^2$ for some constant $C$, and asymptotically
that $\liminf_{n\to\infty}\Var(2n\V{n,k,P})\geq\Var(\chi^2_{k-1})
=2(k-1)$. To the best of our knowledge, no bounds on the higher
moments have appeared in the literature.

In \cref{cor:moments} we showed for every $m\geq 1$ that
$\E{\of{2n\V{n,k,P}}^m}\leq \of{Cm(k-1)}^m$ for some universal constant
$C>0$, and we showed in \cref{thm:asymptotic} the asymptotic equality
$\lim_{n\to\infty}\E{\of{2n\V{n,k,P}}^m}
=2^m\frac{\Gamma\of{m+\frac{k-1}2}}{\Gamma\of{\frac{k-1}2}} =\of{C'
m(k-1)}^m$ where $C'$ is bounded in a constant range.  Thus, our
finite-sample bound is asymptotically optimal up to the universal
constant $C$.

However, the situation is quite different for the central moments
$\E{\of{2n\V{n,k,P}-\E{2n\V{n,k,P}}}^m}$, where we again showed the
finite sample bound $\of{Cm(k-1)}^m$, but asymptotically from
\cref{thm:asymptotic} the bound is $\of{C'm(k-1)}^{\floor{m/2}}$ for
$m\geq 2$ and some $C'$ in a constant range. For $m=2$,
\cite{mar_jia_tan_now_wei_concentration_2019} were able to achieve this
bound up to constant factors, but it is an intriguing open question to
get finite sample central moment bounds with the asymptotically correct
power for $m>2$.

\subsection{Tail bound}

To understand our tail bound (\cref{intro_thm:tail}), we compare our
result to existing bounds in the literature.  Antos and Kontoyiannis
\cite{ant_kon_convergence_2001} used McDiarmid's bounded differences
inequality \cite{mcdiarmid_method_1989} to give a concentration bound
for the empirical entropy, which in the case of the uniform distribution
implies the bound
\[ \PR{\abs{\V{n,k,U_k} - \E{\V{n,k,U_k}}} \geq \eps}
\leq 2e^{-n\eps^2/(2\log^2 n)}. \]
This bound has the advantage of providing subgaussian concentration
around the expectation, but for the case of small $\eps < 1$ it is
preferable to have a bound with linear dependence on $\eps$.
Unfortunately, existing tail bounds which decay like $e^{-n\eps}$ are
not, in the common regime of parameters where $n\gg k$, meaningful for $\eps$
close to $\E{\V{n,k,P}}\leq (k-1)/n$. For example, the method of types
\cite{csiszar_method_1998} is used to prove the standard bound
\begin{equation}
  \PR{\V{n,k,P} > \eps} \leq e^{-n\eps}\cdot\binom{n+k-1}{k-1},
  \label{eqn:mot}
\end{equation}
which is commonly used in proofs of Sanov's theorem (see
e.g.~\cite{cov_tho_elements_2006}).
However, this bound is meaningful only for  $\eps > \frac1n \cdot
\log\binom{n+k-1}{k-1} \geq \frac{k - 1}{n} \cdot \log\of{1 +
\frac{n}{k-1}}$, which is off by a factor of order
$\log\of{1+\frac{n}{k-1}}$. A recent bound due to Mardia et al.\ 
\cite{mar_jia_tan_now_wei_concentration_2019} improved on the method of
types bound for all settings of $k$ and $n$, but for $3\leq k \leq
\frac{e^2}{2\pi}\cdot n$ still requires $\eps > \frac{k}{n}\cdot
\log\of{\sqrt{\frac{e^3n}{2\pi k}}} > \frac{k - 1}{n}\cdot \log\of{1 +
\frac{n-1}{k}}/2$, which again has dependence on $\log\of{1 +
\frac{n-1}{k}}$.

Thus, if $k\leq n$, then our bound is meaningful for $\eps$
smaller than what is needed for the method of types bound or the bound
of \cite{mar_jia_tan_now_wei_concentration_2019} by a factor of order
$\log(n/k)$, which for $k$ as large as $n^{0.99}$ is still $\log(n)$,
and for $k$ as large as $n/\log n$ is of order $\log\log n$.  However,
\cref{intro_thm:tail} has slightly worse dependence on $\eps$ than the
other bounds, so for example it is better than the method of types bound
if and only if
\begin{equation}
  \label{eqn:vsmot}
  \frac{k-1}{n} < \eps < \frac{k-1}{n} \cdot
  \of{\frac 1e\sqrt[k-1]{\binom{n+k-1}{k-1}}}.
\end{equation}
In particular, when $n \geq e(k-1)$, our bound is better for $\eps$ up
to order $\frac{n}{k-1}$ times larger than $\frac{k-1}{n}$.
However, we can also see that our bound can be better only when
$\sqrt[k-1]{\binom{n+k-1}{k-1}}\geq e$, which asymptotically is
equivalent to $k - 1 \leq C n$, where $C\approx 1.84$ is the solution
to the equation $(1+C)/C \cdot H(C/(1+C))=1$ for $H$ the binary entropy
function in nats. From a finite-sample perspective, note that the condition
is always satisfied in the
standard setting of parameters where $n \geq k$, that is, the number
of samples is larger than the size of the alphabet. In this regime, we
can also compare to the ``interpretable'' upper bound of \cite[Theorem
3]{mar_jia_tan_now_wei_concentration_2019}, to see that
\cref{intro_thm:tail} is better if
\[
    \frac{k-1}{n} <  \eps < \frac{k-1}{n}\cdot
    \frac1e\of{\frac{6e^2}{\pi^{3/2}}\sqrt{\frac{e^3n}{2\pi k}}^k
  }^{1/(k-1)},
\]
so that in particular our bound is better for $\eps$ up to order
$\sqrt{\frac nk}^{1+1/(k-1)}\geq \sqrt{\frac nk}$ times larger than
$\frac{k-1}{n}$.

\section{Acknowledgements}

The author thanks Flavio du Pin Calmon and the anonymous reviewers for
their helpful comments and suggestions.



\begin{thebibliography}{10}
\providecommand{\url}[1]{#1}
\csname url@samestyle\endcsname
\providecommand{\newblock}{\relax}
\providecommand{\bibinfo}[2]{#2}
\providecommand{\BIBentrySTDinterwordspacing}{\spaceskip=0pt\relax}
\providecommand{\BIBentryALTinterwordstretchfactor}{4}
\providecommand{\BIBentryALTinterwordspacing}{\spaceskip=\fontdimen2\font plus
\BIBentryALTinterwordstretchfactor\fontdimen3\font minus
  \fontdimen4\font\relax}
\providecommand{\BIBforeignlanguage}[2]{{%
\expandafter\ifx\csname l@#1\endcsname\relax
\typeout{** WARNING: IEEEtran.bst: No hyphenation pattern has been}%
\typeout{** loaded for the language `#1'. Using the pattern for}%
\typeout{** the default language instead.}%
\else
\language=\csname l@#1\endcsname
\fi
#2}}
\providecommand{\BIBdecl}{\relax}
\BIBdecl

\bibitem{pitman_basic_1979}
E.~J.~G. Pitman, \emph{Some Basic Theory for Statistical Inference}, ser.
  Monographs on Applied Probability and Statistics.\hskip 1em plus 0.5em minus
  0.4em\relax {London : New York}: {Chapman and Hall ; distributed in the
  U.S.A. by Halsted Press}, 1979.

\bibitem{ney_pea_problem_1933}
J.~Neyman and E.~S. Pearson, ``\BIBforeignlanguage{en}{On the {{Problem}} of
  the {{Most Efficient Tests}} of {{Statistical Hypotheses}}},''
  \emph{\BIBforeignlanguage{en}{Philosophical Transactions of the Royal Society
  A: Mathematical, Physical and Engineering Sciences}}, vol. 231, no. 694-706,
  pp. 289--337, Jan. 1933.

\bibitem{har_tus_information_2012}
P.~Harremo{\"e}s and G.~Tusn{\'a}dy, ``Information divergence is more
  $\chi^2$-distributed than the $\chi^2$-statistics,'' in \emph{2012 {{IEEE
  International Symposium}} on {{Information Theory Proceedings}}}, Jul. 2012,
  pp. 533--537.

\bibitem{wainwright_highdimensional_2019}
M.~Wainwright, \emph{High-Dimensional Statistics: A Non-Asymptotic Viewpoint},
  ser. Cambridge Series in Statistical and Probabilistic Mathematics.\hskip 1em
  plus 0.5em minus 0.4em\relax {Cambridge ; New York, NY}: {Cambridge
  University Press}, 2019, no.~48.

\bibitem{agrawal_samplers_conference}
R.~Agrawal, ``Samplers and {{Extractors}} for {{Unbounded Functions}},'' in
  \emph{Approximation, {{Randomization}}, and {{Combinatorial Optimization}}.
  {{Algorithms}} and {{Techniques}} ({{APPROX}}/{{RANDOM}} 2019)}, ser. Leibniz
  {{International Proceedings}} in {{Informatics}} ({{LIPIcs}}), D.~Achlioptas
  and L.~A. V{\'e}gh, Eds., vol. 145.\hskip 1em plus 0.5em minus 0.4em\relax
  {Dagstuhl, Germany}: {Schloss Dagstuhl\textendash{}Leibniz-Zentrum fuer
  Informatik}, 2019, pp. 59:1--59:21.

\bibitem{dia_wan_cal_san_robustness_2020}
M.~Diaz, H.~Wang, F.~P. Calmon, and L.~Sankar, ``On the {{Robustness}} of
  {{Information}}-{{Theoretic Privacy Measures}} and {{Mechanisms}},''
  \emph{IEEE Transactions on Information Theory}, vol.~66, no.~4, pp.
  1949--1978, Apr. 2020.

\bibitem{paninski_estimation_2003}
L.~Paninski, ``Estimation of {{Entropy}} and {{Mutual Information}},''
  \emph{Neural Computation}, vol.~15, no.~6, pp. 1191--1253, Jun. 2003.

\bibitem{jia_ven_han_wei_maximum_2017}
J.~Jiao, K.~Venkat, Y.~Han, and T.~Weissman, ``Maximum {{Likelihood
  Estimation}} of {{Functionals}} of {{Discrete Distributions}},'' \emph{IEEE
  Transactions on Information Theory}, vol.~63, no.~10, pp. 6774--6798, Oct.
  2017.

\bibitem{csiszar_method_1998}
I.~Csisz{\'a}r, ``The {{Method}} of {{Types}},'' \emph{IEEE Transactions on
  Information Theory}, vol.~44, no.~6, pp. 2505--2523, Oct. 1998.

\bibitem{mar_jia_tan_now_wei_concentration_2019}
J.~Mardia, J.~Jiao, E.~T{\'a}nczos, R.~D. Nowak, and T.~Weissman,
  ``\BIBforeignlanguage{en}{Concentration inequalities for the empirical
  distribution of discrete distributions: Beyond the method of types},''
  \emph{\BIBforeignlanguage{en}{Information and Inference: A Journal of the
  IMA}}, p. iaz025, Nov. 2019.

\bibitem{ant_kon_convergence_2001}
A.~Antos and I.~Kontoyiannis, ``\BIBforeignlanguage{en}{Convergence properties
  of functional estimates for discrete distributions},''
  \emph{\BIBforeignlanguage{en}{Random Structures \& Algorithms}}, vol.~19, no.
  3-4, pp. 163--193, 2001.

\bibitem{bou_lug_mas_concentration_2013}
S.~Boucheron, G.~Lugosi, and P.~Massart,
  \emph{\BIBforeignlanguage{English}{Concentration {{Inequalities}}: {{A
  Nonasymptotic Theory}} of {{Independence}}}}, 1st~ed.\hskip 1em plus 0.5em
  minus 0.4em\relax {Oxford University Press}, Feb. 2013.

\bibitem{wilks_largesample_1938}
S.~S. Wilks, ``\BIBforeignlanguage{EN}{The {{Large}}-{{Sample Distribution}} of
  the {{Likelihood Ratio}} for {{Testing Composite Hypotheses}}},''
  \emph{\BIBforeignlanguage{EN}{The Annals of Mathematical Statistics}},
  vol.~9, no.~1, pp. 60--62, Mar. 1938.

\bibitem{hoeffding_probability_1963}
W.~Hoeffding, ``Probability {{Inequalities}} for {{Sums}} of {{Bounded Random
  Variables}},'' \emph{Journal of the American Statistical Association},
  vol.~58, no. 301, pp. 13--30, 1963.

\bibitem{gra_knu_pat_concrete_1994}
R.~Graham, D.~Knuth, and O.~Patashnik,
  \emph{\BIBforeignlanguage{English}{Concrete {{Mathematics}}: {{A Foundation}}
  for {{Computer Science}}, {{Second Edition}}}}.\hskip 1em plus 0.5em minus
  0.4em\relax {Upper Saddle River, NJ}: {Addison-Wesley}, 1994.

\bibitem{vershynin_highdimensional_2018}
R.~Vershynin, \emph{High-Dimensional Probability: An Introduction with
  Applications in Data Science}, ser. Cambridge Series in Statistical and
  Probabilistic Mathematics.\hskip 1em plus 0.5em minus 0.4em\relax
  {Cambridge}: {Cambridge University Press}, 2018, no.~47.

\bibitem{csi_shi_information_2005}
I.~Csisz{\'a}r and P.~C. Shields, \emph{Information Theory and Statistics: A
  Tutorial}, ser. Foundations and Trends in Communications and Information
  Theory.\hskip 1em plus 0.5em minus 0.4em\relax {Hanover, MA}: {Now
  Publishers}, 2005.

\bibitem{billingsley_convergence_1999}
P.~Billingsley, \emph{Convergence of Probability Measures}, 2nd~ed., ser. Wiley
  Series in Probability and Statistics. {{Probability}} and Statistics
  Section.\hskip 1em plus 0.5em minus 0.4em\relax {New York}: {Wiley}, 1999.

\bibitem{zub_ser_complete_2013}
A.~M. Zubkov and A.~A. Serov, ``A {{Complete Proof}} of {{Universal
  Inequalities}} for the {{Distribution Function}} of the {{Binomial Law}},''
  \emph{Theory of Probability \& Its Applications}, vol.~57, no.~3, pp.
  539--544, Jan. 2013.

\bibitem{harremoes_bounds_2017}
P.~Harremo{\"e}s, ``\BIBforeignlanguage{en}{Bounds on tail probabilities for
  negative binomial distributions},''
  \emph{\BIBforeignlanguage{en}{Kybernetika}}, pp. 943--966, Feb. 2017.

\bibitem{mcdiarmid_method_1989}
C.~McDiarmid, ``On the method of bounded differences,'' in \emph{Surveys in
  {{Combinatorics}}, 1989}, J.~Siemons, Ed.\hskip 1em plus 0.5em minus
  0.4em\relax {Cambridge}: {Cambridge University Press}, 1989, pp. 148--188.

\bibitem{cov_tho_elements_2006}
T.~M. Cover and J.~A. Thomas, \emph{Elements of Information Theory},
  2nd~ed.\hskip 1em plus 0.5em minus 0.4em\relax {Hoboken, N.J}:
  {Wiley-Interscience}, 2006.

\end{thebibliography}
\end{document}